\documentclass[conference]{IEEEtran}
\IEEEoverridecommandlockouts
\usepackage{cite}
\usepackage{amsmath,amssymb,amsfonts}
\usepackage{graphicx}
\usepackage{textcomp}
\usepackage{xcolor}
\usepackage{graphicx}
\usepackage{amsmath}
\usepackage[noend]{algpseudocode}
\usepackage{algorithmicx,algorithm}
\usepackage{amsthm}
\usepackage{amsfonts}
\usepackage{subfigure} 
\usepackage{color}
\usepackage{cite}
\usepackage{setspace}
\newtheorem{Lemma}{Lemma}
\newtheorem{Theorem}{Theorem}
\usepackage{amssymb}
\usepackage{color}
\graphicspath{{images/}}
\usepackage{booktabs}
\usepackage{multirow} 
\usepackage{makecell}
\usepackage{threeparttable}
\usepackage[numbers,sort&compress]{natbib}
\newtheorem{remark}{Remark}
\usepackage{stfloats}
\usepackage{amsmath}
\usepackage{amssymb}
\usepackage{enumitem}
\usepackage{comment}
\usepackage{bm}
\setlist{noitemsep, topsep=2pt, parsep=0pt, partopsep=0pt}
\def\BibTeX{{\rm B\kern-.05em{\sc i\kern-.025em b}\kern-.08em
    T\kern-.1667em\lower.7ex\hbox{E}\kern-.125emX}}
\begin{document}

\title{\huge Sea-State-Induced Performance Transition in Maritime Networks: A Roughness-Aware Stochastic Geometry Framework}

\author{Wen-Yu Dong, Shaoshi Yang,~\IEEEmembership{Senior Member,~IEEE}, Song Zhao, Rui-Si Han, \\ Qi Bi,~\IEEEmembership{Fellow,~IEEE}, and Sheng Chen,~\IEEEmembership{Life Fellow,~IEEE}	%

\thanks{W.-Y. Dong, S. Zhao and Q. Bi are with the Future Technology Research Center, China Telecom Research Institute, Beijing 102209, China (Email: dongwy@chinatelecom.cn; zhaosong1@chinatelecom.cn; qibi@chinatelecom.cn)}
\thanks{S. Yang is with the School of Information and Communication Engineering, Beijing University of Posts and Telecommunications, and also with the Key Laboratory of Mathematics and Information Networks, Ministry of Education, Beijing 100876, China (E-mail:  shaoshi.yang@bupt.edu.cn).}
\thanks{R.-S. Han is with the Cloud Network Operating System R\&D Center, China Telecom, Beijing 102209, China (Email: hanruisi@chinatelecom.cn)}
\thanks{S. Chen is with the School of Electronics and Computer Science, University of Southampton, Southampton SO17 1BJ, U.K. (E-mail: sqc@ecs.soton.ac.uk).} %
\vspace*{-3mm}

\thanks{\scriptsize This manuscript presents an early version of our study on the impact of sea surface roughness on near-shore maritime networks. A substantially extended version has been published in \emph{IEEE Transactions on Communications} under the title ``Proper Sea Surface Roughness Enhances the Performance of Near-Shore Maritime Networks,'' with additional theoretical analysis, refined propagation modeling, and more comprehensive performance evaluations.}}
\maketitle

\begin{abstract}
Analytical studies of maritime wireless networks commonly assume a deterministic smooth sea surface, leaving unclear how realistic ocean conditions reshape network-level reliability. Unlike conventional intuition that sea roughness always deteriorates propagation, this work reveals a non-monotonic sea-state-induced performance transition caused by the competition between interference-null mitigation and coherent reflection loss. This paper develops a physically grounded, sea-state-aware stochastic geometry framework for maritime networks by incorporating sea surface roughness into propagation modeling. Specifically, we derive an effective reflection coefficient based on the classical Rayleigh roughness criterion, where the significant wave height explicitly characterizes the attenuation of the coherent specular reflection component caused by surface roughness. By integrating the proposed channel model into a stochastic geometry framework, we derive tractable expressions for uplink coverage probability under different sea states. Our analysis reveals a non-monotonic impact of sea roughness on network performance under the considered propagation model: moderate roughness can improve reliability-oriented coverage by mitigating destructive interference nulls, whereas stronger roughness attenuates coherent reflected energy and degrades high-SINR performance. Measurement comparisons support the underlying roughness-sensitive reflection mechanism, while rough-sea VHF results are interpreted as wavelength-specific model predictions rather than direct empirical validation.
\end{abstract}

\begin{IEEEkeywords}
Maritime communication, physical channel model, sea surface roughness, stochastic geometry, coverage probability.
\end{IEEEkeywords}

\section{Introduction}\label{S1}

The performance of these vital systems is, however, uniquely constrained by the propagation channel over the sea, where electromagnetic propagation is strongly affected by surface reflection and ocean conditions \cite{Alqurashi2023, Lyu2021}. Unlike terrestrial environments with relatively stable propagation characteristics, maritime channels exhibit distinct propagation behaviors due to the interaction between electromagnetic waves and sea surfaces. Therefore, developing physically accurate channel models that capture the impact of sea conditions on network-level performance is a fundamental requirement for robust maritime network analysis.

However, a significant gap exists between the physical reality of the maritime channel and the models predominantly used in theoretical network analysis. The literature on this topic has evolved along two primary, yet largely disconnected, trajectories. The first focuses on high-fidelity physical channel modeling, employing computationally intensive methods like stochastic ray tracing \cite{Ding2019} or geometric stochastic channel models \cite{Raulefs2023} to capture detailed link-level physics. While these models offer exceptional accuracy for single links, their inherent complexity makes them intractable for the performance analysis of large-scale networks with randomly located users.

The second trajectory leverages tractable network-level analysis, with stochastic geometry emerging as a powerful tool for analyzing interference-limited wireless systems \cite{DongTCOM, DongJSAC, DongIoTJ}. Recently, fluid-spatiotemporal stochastic geometry (F-STSG) has further extended stochastic geometric modeling toward non-stationary wireless environments by characterizing the evolution of information flows over spatial-temporal fields \cite{Dong2026FSTSG}. However, to maintain tractability, its application to maritime or integrated networks has resorted to overly simplistic models that lack physical fidelity. For instance, recent shore-to-ship analyses either adopt a standard power-law path loss \cite{Hu2024}, or employ empirical ITU models derived from curve-fitting \cite{Xu2023}. These simplifications entirely neglect the dominant multipath reflection and atmospheric ducting phenomena that define the maritime channel. Even the seminal, measurement-validated piecewise model by Lee et al. \cite{Lee2014}, a cornerstone in the field, is predicated on an idealized, perfectly smooth sea surface and was validated only under calm conditions, failing to physically parameterize the impact of different sea states on reflection characteristics. This analytical gap is compounded by equally simplistic spatial models, which often assume a single representative link or a uniform vessel distribution \cite{Xu2023, Hu2024}, failing to capture the non-uniform, arbitrary patterns of real-world maritime traffic. Our previous work introduced the coverage probability density concept to characterize spatially resolved performance variations in non-homogeneous maritime networks, highlighting the importance of realistic vessel distributions in network-level analysis \cite{Dong2026CPDF}. Consequently, a significant disconnect persists between the physical realities of the near-shore environment and the theoretical frameworks used for system-level performance analysis.

More importantly, existing maritime network analyses usually treat sea surface conditions as fixed assumptions and rarely quantify how different sea states affect stochastic network performance. This creates a gap between physical channel realism and analytical network modeling. Bridging this gap requires a framework that jointly incorporates realistic sea-surface propagation mechanisms and tractable stochastic geometry analysis. This perspective is also consistent with our recent efforts on fluid-based modeling of non-stationary wireless systems, where environmental and traffic variations are incorporated into network-level analysis for infrastructure provisioning and resource adaptation \cite{DongDigitalTides2026}.

This paper bridges this critical gap by making the following novel contributions:
\begin{itemize}
\item We establish a sea-state-dependent propagation framework that explicitly connects environmental conditions to network-level performance. By deriving an effective reflection coefficient from the Rayleigh roughness criterion, we reveal how significant wave height, carrier wavelength, and grazing angle jointly determine the transition from coherent multipath interference to rough-surface scattering.
\item We integrate the proposed physical channel model into a tractable stochastic geometry framework and derive analytical expressions for uplink coverage probability under roughness-aware maritime propagation conditions.
\item Our framework reveals a reliability--rate trade-off induced by sea
surface roughness. Specifically, moderate roughness can improve
reliability-oriented coverage by mitigating destructive interference
nulls, whereas stronger roughness attenuates the coherent reflected
component and degrades high-SINR transmission performance.
\end{itemize}

\section{System Model}\label{S2}

\subsection{Spatial and Network Model}\label{S2.1}

We consider a near-shore network with a single onshore station (OS) at the origin and vessels distributed on a 2D half-plane. To capture realistic traffic patterns, their locations are modeled by a non-homogeneous Poisson point process (NHPPP), $\Phi_S$, with an isotropic intensity function 
\begin{align}\label{eqNHPPP} 
  \mu(d) = \mu_0 d^{\alpha-1}e^{-\beta d},
\end{align}
where $\mu_0$ is the intensity constant.
This specific function is chosen because it effectively models realistic maritime traffic patterns where vessel density may peak in shipping lanes away from the coast. Specifically, the combination of a power-law term, i.e., $d^{\alpha-1}$, and an exponential decay term, i.e., $e^{-\beta d}$, allows the vessel density to be very low near the coast, peak at a certain distance offshore corresponding to a main shipping lane, and then diminish further out at sea. The parameters $\alpha > 1$ and $\beta$ provide the flexibility to control the location and concentration of this peak density. The term ``isotropic'' signifies that the density depends only on the distance $d$ from the OS, not the angle, simplifying the model to represent shipping lanes as a high-density annulus.

For uplink medium access, we adopt an ALOHA-like protocol where each vessel transmits independently with a probability $p_a$. Consequently, the set of concurrently transmitting interferers, $\Phi'_S$, also forms a thinned NHPPP with intensity 
\begin{align}\label{eqSNHPPP} 
  \mu'(d) = p_a \mu(d).
\end{align}

\subsection{Channel and SINR Model}\label{S2.2}

The small-scale fading power gain, $|h|^2$, for any given link is assumed to follow a Nakagami-$m$ distribution. This model is chosen for its tractability and flexibility in accurately representing the line-of-sight (LOS) dominant characteristics of the maritime channel.

\begin{figure}[!t]
\begin{center}
	\includegraphics[width=\columnwidth]{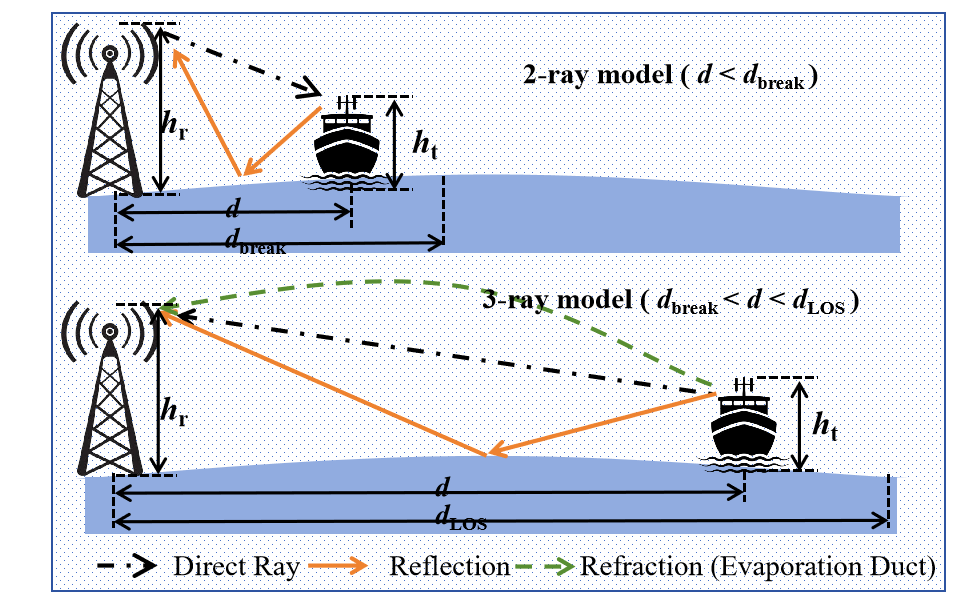}
\end{center}
\vspace*{-4mm}
\caption{Illustration of the piecewise path loss framework.}
\label{fig:system_model_illustration} 
\vspace*{-3mm}
\end{figure}

The performance of a typical vessel at a distance $D$ from the OS is determined by the uplink signal-to-interference-plus-noise ratio (SINR), defined as:
\begin{equation}\label{eqSINR} 
	\mathrm{SINR} = \frac{P_{S_m}|h|^2 G(D)}{I + N_0},
\end{equation}
where $P_{S_m}$ is the transmit power of the typical vessel, $G(D) = 1/PL(D)$ is the large-scale path gain, $I = \sum_{S_i \in \Phi'_S} P_{S_i}|h_i|^2 G(D_i)$ is the aggregate interference from other active vessels, with $P_{S_i}$ being the transmit power of an interfering vessel, and $N_0$ is the additive white Gaussian noise power.

\subsection{Baseline Large-Scale Path Loss Model}\label{S2.3}

For the large-scale path gain, we adopt the widely-used piecewise path loss framework \cite{Lee2014} as our baseline, which is predicated on an ideal, perfectly smooth sea surface ($\Gamma \approx -1$). As illustrated in Fig.~\ref{fig:system_model_illustration}, this model combines two propagation regimes based on the distance $d$ relative to a breakpoint $d_{\mathrm{break}}\! =\! 4h_{\textrm t} h_{\textrm r} / \lambda$, where $h_{\textrm t}$ and $h_{\textrm r}$ are the transmitting and receiving antenna heights, respectively, while $\lambda$ is the carrier wavelength.

\subsubsection{Two-Ray Model ($d < d_{\mathrm{break}}$)}
At shorter distances, propagation is accurately described by the coherent sum of the direct and sea-reflected paths, as shown in the top panel of Fig.~\ref{fig:system_model_illustration}. The path loss is given by:
\begin{equation}\label{eqTwoRM} 
	PL_{\mathrm{2-ray}}^{\mathrm{ideal}}(d) = \frac{(4\pi d/\lambda)^2}{\big(2\sin\big(\frac{2\pi h_{\textrm t} h_{\textrm r}}{\lambda d}\big)\big)^2}.
\end{equation}

\subsubsection{Three-Ray Model ($d \geq d_{\mathrm{break}}$)}
Beyond the breakpoint, a third path resulting from reflection and refraction within the atmospheric evaporation duct becomes significant, as depicted in the bottom panel of Fig.~\ref{fig:system_model_illustration}. The path loss is modeled as:
\begin{equation}\label{eqThreeRM} 
	PL_{\mathrm{3-ray}}^{\mathrm{ideal}}(d) = \frac{(4\pi d/\lambda)^2}{(2(1+\Delta(d)))^2},
\end{equation}
where the interference term $\Delta(d)$ is given by:
\begin{equation}\label{eqInfTerm} 
	\Delta(d) = 2\sin\left(\frac{2\pi h_{\textrm t} h_{\textrm r}}{\lambda d}\right)\sin\left(\frac{2\pi(h_{\textrm e}-h_{\textrm t})(h_{\textrm e}-h_{\textrm r})}{\lambda d}\right),
\end{equation}
with $h_{\textrm e}$ being the effective height of the evaporation duct.

\section{The Roughness-Aware Path Loss Model}\label{S3}

\subsection{Derivation of the Effective Reflection Coefficient}\label{S3.1}

We quantify sea state using the significant wave height ($H_s$). The electromagnetic impact of roughness is governed by the Rayleigh roughness factor $\xi$ \cite{Pinel2010}:
\begin{equation}\label{eq:rayleigh_factor} 
	\xi = \frac{4\pi H_{\text{s}} \sin\psi}{\lambda},
\end{equation}
where $\psi\! \approx\! (h_{\text{t}} + h_{\text{r}}) / d$ is the grazing angle. This random phase variation attenuates the coherent specular component by a factor $\rho_{\text{spec}}$, which for a Gaussian surface height distribution is given by:
\begin{equation}\label{eqGSH} 
	\rho_{\text{spec}} = \exp\left(-\frac{\xi^2}{2}\right).
\end{equation}
The effective reflection coefficient $\rho_s$ is the product of the ideal coefficient $\Gamma$ and this scattering factor:
\begin{equation}\label{eqERC} 
	\rho_{\text{s}} = \Gamma \cdot \rho_{\text{spec}} \approx -1 \cdot \exp\left(-\frac{\xi^2}{2}\right).
\end{equation}
For a smooth sea ($H_s \to 0$), $\rho_s \to -1$, which recovers the ideal model. For rough seas ($H_s > 0$), $|\rho_s| < 1$, which correctly models the reduction in specular power.

\subsection{Modified Path Loss Equations}\label{S3.2}

\begin{Theorem}\label{T1}
For a sea surface with a non-ideal reflection coefficient $\rho_{\mathrm{s}}$, the linear path loss, $P\!L^{\rm{rough}}(d)$, is given by a piecewise function.
\begin{itemize}
\item For the two-ray region ($d < d_{\mathrm{break}}$):
	\begin{equation} \label{eq:rough_2ray} 
		P\!L_{\rm{2-ray}}^{\rm{rough}}(d) = \frac{(4\pi d / \lambda)^2}{1 + \rho_{\mathrm{s}}^2 - 2|\rho_{\mathrm{s}}|\cos\left(\frac{4\pi h_{\mathrm{t}} h_{\mathrm{r}}}{\lambda d}\right)}.
	\end{equation}
\item For the three-ray region ($d \geq d_{\mathrm{break}}$):
	\begin{equation}\label{eq:rough_3ray} 
		P\!L_{\rm{3-ray}}^{\mathrm{rough}}(d)\! =\! \frac{(4\pi d / \lambda)^2}{(1\! +\! \rho_{\mathrm{s}} C_2\! -\! C_3)^2 + (\rho_{\mathrm{s}} S_2\! -\! S_3)^2},
	\end{equation}
\end{itemize}
where $C_i=\cos(k\Delta d_i)$ and $S_i=\sin(k\Delta d_i)$, $i=2,3$, with $\Delta d_2$ and $\Delta d_3$ being the respective path length differences.
\end{Theorem}

\begin{proof}
The proof is derived from the coherent superposition of the electric field components for the two-ray and three-ray models, respectively.

\textit{1) Two-Ray Model ($d < d_{\mathrm{break}}$):}
The total received electric field, $E_{\text{total}}$, is the vector sum of the LOS field, $E_{\text{LOS}}$, and the sea-reflected field, $E_{\text{refl}}$:
\begin{equation}\label{eqTwoMef} 
	E_{\text{total}} = E_{\text{LOS}} + E_{\text{refl}} = \frac{E_0}{d_1} e^{-jkd_1} + \rho_{\text{s}} \frac{E_0}{d_2} e^{-jkd_2},
\end{equation}
where $E_0$ is the initial field strength, $d_1$ and $d_2$ are the path lengths of the direct and reflected rays, respectively, and $k=2\pi/\lambda$ is the wavenumber.
For far-field scenarios ($d \gg h_{\text{t}}, h_{\text{r}}$), we use the standard approximations: $d_1 \approx d_2 \approx d$ for the amplitude, and a Taylor expansion for the path length difference, $\Delta d = d_2 - d_1 \approx \frac{2 h_{\text{t}} h_{\text{r}}}{d}$, for the phase. The total field can then be written as:
\begin{equation}\label{eqTwoMef1} 
	E_{\text{total}} \approx \frac{E_0}{d} e^{-jkd_1} \left( 1 + \rho_{\text{s}} e^{-j k\Delta d} \right).
\end{equation}
The received power is proportional to the squared magnitude of this field. Since $|e^{-jkd_1}|\! =\! 1$ and $\rho_{\text{s}}$ is a negative real number, the multipath power gain factor is given by:
\begin{align}\label{eqMPG2} 
	G_{\text{mp}}^{\text{2-ray}} =& 1 + |\rho_{\text{s}}|^2 - 2|\rho_{\text{s}}|\cos(k\Delta d) \nonumber \\
	=& 1 + \rho_{\text{s}}^2 - 2|\rho_{\text{s}}|\cos(k\Delta d).
\end{align}
The path loss is the free-space path loss, $P\!L_{\text{fs}} = (4\pi d / \lambda)^2$, divided by this gain factor, which yields (\ref{eq:rough_2ray}).
	
\textit{2) Three-Ray Model ($d \geq d_{\mathrm{break}}$):}
This model adds a third path component, $E_{\text{duct}}$, resulting from refraction by the evaporation duct, which is treated as a reflection from an effective height $h_{\textrm e}$ with a reflection coefficient of $-1$. Hence,
\begin{equation}\label{eqThrMef1} 
	E_{\text{total}} \approx \frac{E_0}{d} \left( e^{-jkd_1} + \rho_{\text{s}} e^{-jkd_2} - e^{-jkd_3} \right) ,
\end{equation}
where the path length difference for the duct-refracted ray is $\Delta d_3\! =\! d_3\! -\! d_1\! \approx\! \frac{2(h_{\mathrm{e}} - h_{\mathrm{t}})(h_{\mathrm{e}} - h_{\mathrm{r}})}{d}$. Thus the multipath gain is:
\begin{align}\label{eqMPG3} 
	G_{\text{mp}}^{\text{3-ray}} &= (1 + \rho_{\text{s}} C_2 - C_3)^2 + (\rho_{\text{s}} S_2 - S_3)^2,
\end{align}
where $C_i=\cos(k\Delta d_i)$ and $S_i=\sin(k\Delta d_i)$, $i=2,3$. This completes the proof.
\end{proof}

\begin{remark}\label{remark:limits}
	The proposed roughness-aware model naturally unifies two extreme propagation regimes. 
	When $H_s\rightarrow0$, the Rayleigh factor satisfies $\rho_{\rm spec}\rightarrow1$, and hence $\rho_s\rightarrow-1$, recovering the conventional smooth-sea two-ray/three-ray model. 
	Conversely, when $H_s\rightarrow\infty$, $\rho_s\rightarrow0$, indicating the disappearance of coherent sea reflection and reducing the model toward a dominant direct-path propagation regime.
	Therefore, the proposed model provides a continuous transition between smooth-sea multipath propagation and rough-sea scattering propagation.
\end{remark}

\section{Sea-State-Aware Uplink Coverage Analysis}
\label{S4}

Unlike conventional stochastic geometry analyses that assume a fixed propagation
environment, the proposed framework incorporates the environmental variable
$H_s$ into both desired signal and interference links through the roughness-aware
path loss model developed in Section~III. Therefore, the uplink coverage
probability becomes a function of both the SINR threshold and sea state,
denoted by $P_{\mathrm{cov}}^{\mathrm{UL}}(T,H_s)$.

\subsection{Interference Characterization under Roughness-Aware Propagation}
\label{S4.1}

\begin{Lemma}\label{L1}
	Under the proposed sea-state-dependent propagation model, the Laplace transform
	of the aggregate interference can be decomposed into the product of the
	near-field and far-field interference components:	
	\begin{equation}
		\mathcal{L}_{I}(s,H_s)
		=
		\mathcal{L}_{I_{\rm near}}(s,H_s)
		\mathcal{L}_{I_{\rm far}}(s_1,H_s).
	\end{equation}
	
	The corresponding transforms are given by	
	\begin{align}
		\mathcal{L}_{I,\mathrm{near}}(s,H_s)
		=& \!
		\exp\Bigg(\!\!
		\pi \!\!
		\int_0^{d_{\mathrm{break}}}  \!\!\!
		\left[ \!\!
		\left(
		1  \!+  \!
		\frac{sP_{S_i}}
		{mP\!L_{\rm 2-ray}^{\rm rough}(d_i,H_s)}
		\right)^{\!\!\!-m}  \right. \nonumber\\
		&
		-1
		\bigg]
		\times
		p_a\mu_0 d_i^{\alpha}
		e^{-\beta d_i}
		\,\mathrm{d}d_i
		\Bigg),
	\end{align}

	\begin{align}
		\mathcal{L}_{I,\mathrm{far}}(s_1,H_s)
		=&\!
		\exp\Bigg(\!\!
		\pi\!\!
		\int_{d_{\mathrm{break}}}^{d_{\mathrm{LOS}}}\!\!\!
		\left[\!\!
		\left(\!
		1 \!+\!
		\frac{s_1P_{S_i}}
		{mPL_{\rm 3-ray}^{\rm rough}(d_i,H_s)}
		\right)^{\!\!\!-m} \right. \nonumber\\
		&
		-1
		\bigg]
		\times
		p_a\mu_0 d_i^{\alpha}
		e^{-\beta d_i}
		\,\mathrm{d}d_i
		\Bigg),
	\end{align}
		where $p_a$ is the medium access probability,
	$s=-n\eta TPL_{\rm 2-ray}^{\rm rough}(d,H_s)/P_{S_m}$,
	and
	$s_1=-n\eta TPL_{\rm 3-ray}^{\rm rough}(d,H_s)/P_{S_m}$.
\end{Lemma}

\begin{proof}
	
	Following the definition of the interference Laplace transform,	
	\begin{align*}
		\mathcal{L}_{I,\mathrm{near}}(s,H_s)
		=&
		\mathbb{E}
		\left[
		\exp
		\left(
		-s
		\sum_{S_i\in\Phi'_S}
		\frac{P_{S_i}|h_i|^2}
		{PL_{\rm 2-ray}^{\rm rough}(D_i,H_s)}
		\right)
		\right].
	\end{align*}

	By applying the moment generating function of the normalized Gamma	random variable,
	
	\begin{align*}
		=&
		\mathbb{E}_{\Phi'_S}
		\left[
		\prod_{S_i\in\Phi'_S}
		\left(
		1+
		\frac{sP_{S_i}}
		{mPL_{\rm 2-ray}^{\rm rough}(d_i,H_s)}
		\right)^{-m}
		\right].
	\end{align*}

	Using the probability generating functional (PGFL) of the NHPPP,
		\begin{align}
		=&
		\exp
		\left(
		\int_S
		\left[
		\left(
		1+
		\frac{sP_{S_i}}
		{mPL_{\rm 2-ray}^{\rm rough}(d_i,H_s)}
		\right)^{-m}
		\!\!\!\!\!\!-1
		\right]
		\mu'(d_i)
		\mathrm{d}{\bf x}
		\right).
	\end{align}

	By exploiting the radial symmetry of the spatial intensity function,	
	\[
	\mathrm{d}{\bf x}
	=
	\pi d_i\mathrm{d}d_i ,
	\]
	which yields the near-field expression.
	
	The far-field interference transform follows the same procedure by replacing
	the two-ray propagation term with the roughness-aware three-ray propagation
	term.
	
\end{proof}

\subsection{Coverage Probability under Sea-State Variation}
\label{S4.2}

\begin{Theorem}\label{The2}
	
	For a maritime network operating under sea state $H_s$, the uplink coverage
	probability is expressed as
	
	\begin{align}
		P_{\mathrm{cov}}^{\mathrm{UL}}(T,H_s)
		=&
		\int_0^{d_{\rm break}}
		\sum_{n=1}^{m}
		(-1)^{n+1}
		\binom{m}{n}
		\exp(-sN_0)
		\nonumber\\
		&\times
		\mathcal{L}_{I,\mathrm{near}}(s,H_s)
		g(d)\mathrm{d}d
		\nonumber\\
		&
		+
		\int_{d_{\rm break}}^{d_{\rm LOS}}
		\sum_{n=1}^{m}
		(-1)^{n+1}
		\binom{m}{n}
		\exp(-s_1N_0)
		\nonumber\\
		&\times
		\mathcal{L}_{I,\mathrm{far}}(s_1,H_s)
		g(d)\mathrm{d}d ,
	\end{align}
	where
	$
	s=
	-\frac{n\eta T
		PL_{\rm 2-ray}^{\rm rough}(d,H_s)}
	{P_{S_m}},
$
	and
	$
	s_1=
	-\frac{n\eta T
		PL_{\rm 3-ray}^{\rm rough}(d,H_s)}
	{P_{S_m}}.
	$
	\end{Theorem}

\begin{proof}
		The coverage probability is defined as
$
	P_{\mathrm{cov}}^{UL}(T,H_s)
	=
	\mathbb{P}
	(\mathrm{SINR}_{S_m}\geq T).
$

	Conditioning on the serving distance gives
		\begin{align}
		P_{\mathrm{cov}}^{UL}
		=&
		\int_0^{d_{\rm break}}
		\mathbb{P}
		(\mathrm{SINR}^{\mathrm{2-ray}}_{S_m}\geq T|d)
		g(d)\mathrm{d}d
		\nonumber\\
		&
		+
		\int_{d_{\rm break}}^{d_{\rm LOS}}
		\mathbb{P}
		(\mathrm{SINR}^{\mathrm{3-ray}}_{S_m}\geq T|d)
		g(d)\mathrm{d}d .
	\end{align}

	For Nakagami-$m$ fading, applying the Alzer inequality gives
		\begin{align}
		\mathbb{P}
		(\mathrm{SINR}_{S_m}^{\mathrm{2-ray}}\!\!\geq \!T|d)
		\!=\!
		\sum_{n=1}^{m} \!
		(-1)^{n+1} \!
		\binom{m}{n}
		\mathcal{L}_{I,\mathrm{near}}(s,H_s)
		e^{-sN_0}.
	\end{align}

	Similarly,
	\begin{align}
		\mathbb{P}
		(\mathrm{SINR}_{S_m}^{\mathrm{3-ray}}\!\!\geq\! T|d)
		\!=\!
		\sum_{n=1}^{m} \!
		(-1)^{n+1} \!
		\binom{m}{n}
		\mathcal{L}_{I,\mathrm{far}}(s_1,H_s)
		e^{-s_1N_0}.
	\end{align}
	Substituting the above expressions completes the proof.
\end{proof}

\subsection{Sea-State-Induced Coverage Transition}
\label{S4.3}

The analytical expression above indicates that sea roughness affects coverage
through two competing mechanisms. For small $H_s$, reducing the coherent
reflection coefficient weakens destructive multipath cancellation and can
improve reliability. However, when the sea surface becomes sufficiently rough,
the attenuation of coherent reflected energy dominates and reduces the useful
received power.

To characterize this transition, we define the optimal significant wave height
for a given reliability requirement as

\begin{equation}
	H_s^\star
	=
	\arg\max_{H_s}
	P_{\mathrm{cov}}^{\mathrm{UL}}(T,H_s).
\end{equation}
The value $H_s^\star$ represents the operating point where the benefit of
multipath-null mitigation is balanced by the loss of coherent reflected power.
Therefore, the proposed framework provides a direct analytical connection
between environmental sea conditions and network-level reliability.

\begin{remark}\label{remark:sea_state_limits}
	The proposed model includes the conventional smooth-sea model as a special case. Specifically, when $H_s\rightarrow0$, the roughness attenuation factor approaches one, i.e., $\rho_{\mathrm{spec}}\rightarrow1$, and the effective reflection coefficient becomes $\rho_s\rightarrow-1$, which recovers the coherent two-ray/three-ray propagation model. As $H_s$ increases, the coherent reflected component is gradually attenuated and $\rho_s$ approaches zero under severe roughness conditions. In this case, the propagation behavior becomes increasingly dominated by the direct and non-coherent components. Therefore, the proposed model characterizes the variation of maritime propagation behavior with sea surface roughness while maintaining consistency with existing propagation models in the corresponding limiting cases.
\end{remark}
\section{Numerical Results}\label{S5}

In this section, we validate our derived analytical expressions against Monte Carlo simulations and empirical data. Unless otherwise stated, the simulations assume a carrier frequency of 162\,MHz with a 25\,kHz bandwidth. The vessel transmit power is set to 30\,dBm against a noise power of -174\,dBm. Small-scale fading follows a Nakagami-$m$ distribution with $m=3$. Key geometric parameters include OS antenna height of 50\,m, vessel antenna height of 10\,m, and an evaporation duct height of 30.5\,m. The spatial distribution of vessels is governed by a NHPPP with Gamma parameters $\alpha\! =\! 3$ and $\beta\! =\! 1/500$, a base density of $\mu_0\! =\! 1 \times 10^{-4}$~m$^{-2}$, and a medium access probability of $p_a=0.2$.
  
\begin{figure}[!t]
\vspace*{-1mm}
	\centering
	\subfigure[Two-ray region validation and residuals.]{\includegraphics[width=1\columnwidth]{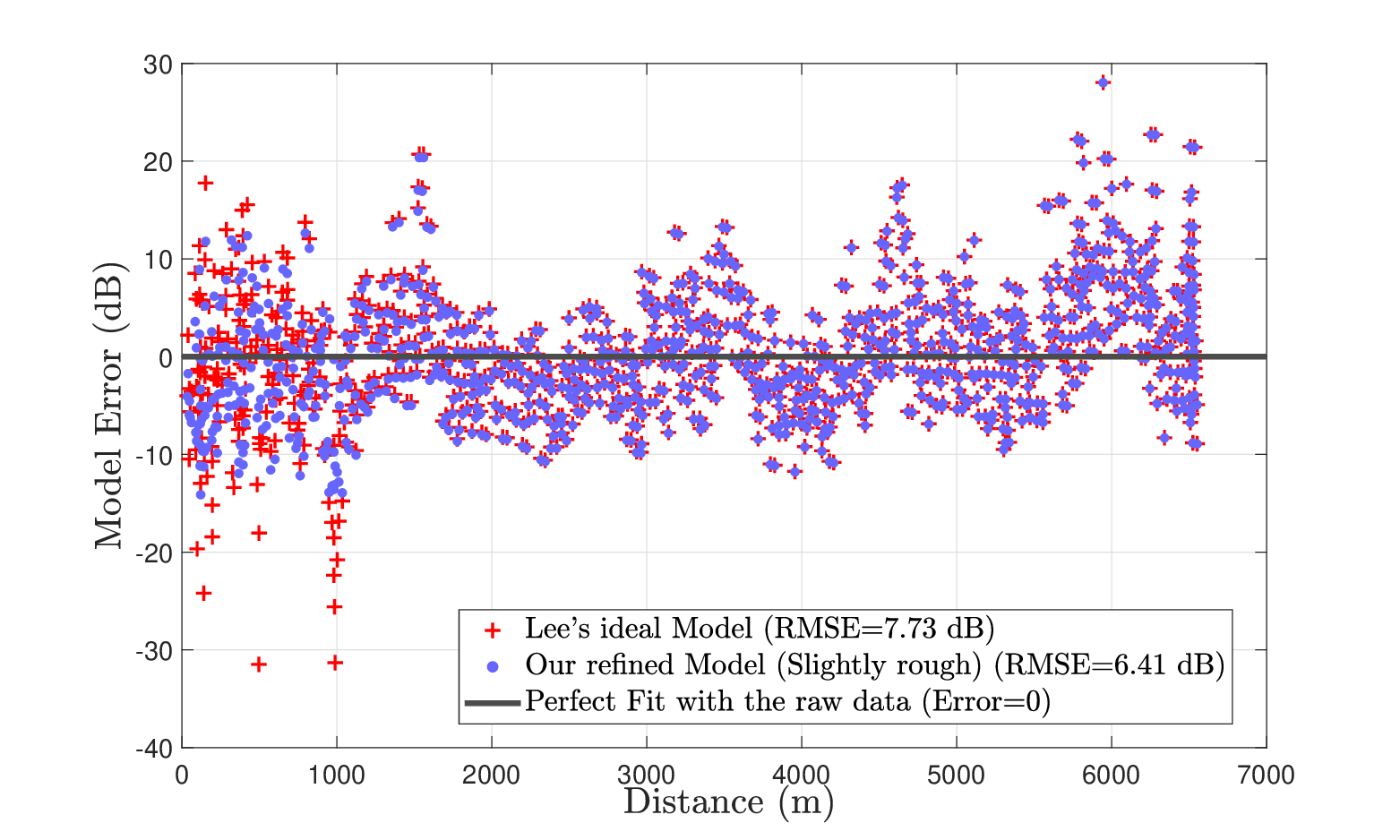}\label{fig:2ray_validation}}
	\subfigure[Three-ray validation and residuals.]{\includegraphics[width=1\columnwidth]{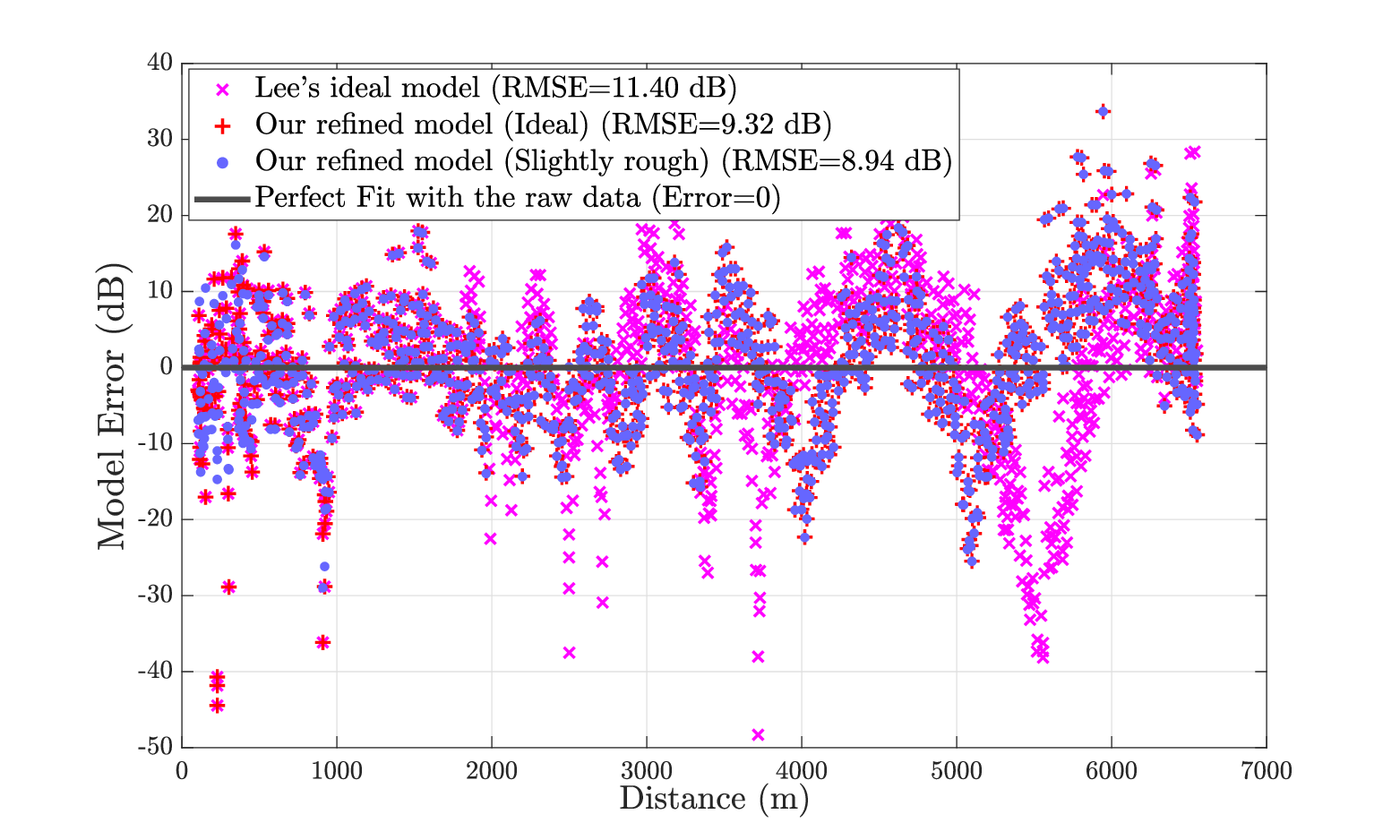}\label{fig:3ray_validation}}
		\vspace*{-5mm}
	\caption{Model validation against measured data from \cite{Lee2014}.}
	\label{fig:validation_combined} 
\vspace*{-4mm}
\end{figure}

\subsection{Model Validation against Empirical Data}\label{S5.1}

We first validate the proposed roughness-aware model against the real-world measurement data from \cite{Lee2014}. Fig.~\ref{fig:2ray_validation} focuses on the two-ray dominant region. The ideal model, predicated on perfect specular reflection, i.e., $\rho_s\! =\! -1$, exhibits deep, periodic nulls that significantly underestimate the measured signal strength within these fades. In contrast, our roughness-aware model, by accounting for the diffuse scattering caused by a sea state of $H_s\! =\! 1$\,m, generates shallower fades that more accurately envelop the scatter of the empirical data points. This qualitative improvement is quantitatively confirmed by a reduction in the root mean square error (RMSE) from 7.73\,dB to 6.41\,dB.

Fig.~\ref{fig:3ray_validation} extends this analysis to the full propagation range where the three-ray model is active. For distances beyond the breakpoint, i.e., $d \geq d_{\mathrm{break}}$, the ideal model exhibits significantly larger deviations and fluctuations compared to the empirical data. Our refined model, however, which correctly incorporates both the atmospheric ducting effect and surface roughness, demonstrates a much closer agreement with the data. This validates the necessity and accuracy of our comprehensive model, which reduces the overall RMSE from 11.40\,dB to 8.94\,dB.

\subsection{Coverage Performance Analysis}\label{S5.2}

\begin{figure}[!t]
	\vspace*{-1mm}
	\begin{center}
		\includegraphics[width=1\columnwidth]{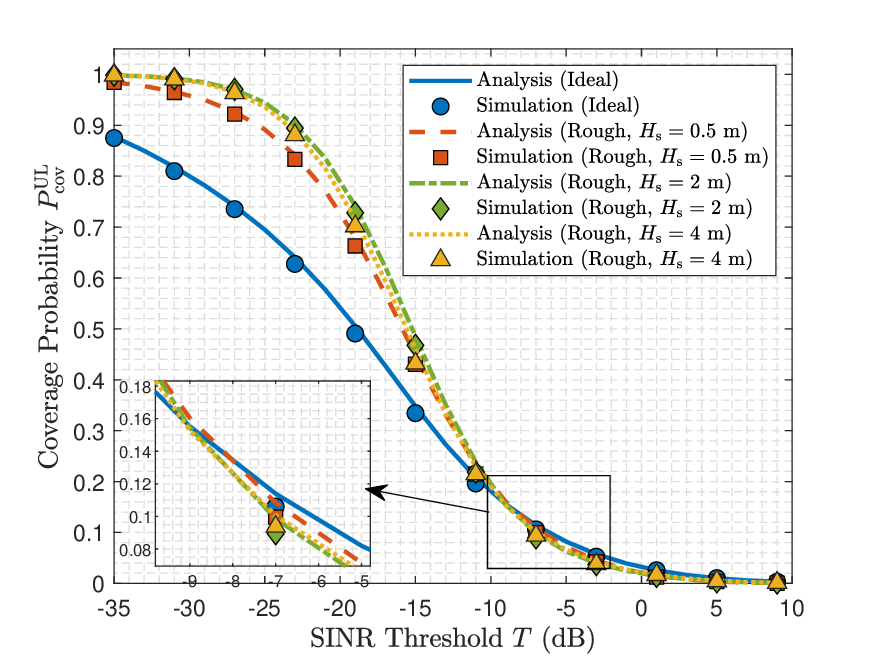}
	\end{center}
	\vspace*{-5mm}
	\caption{Coverage probability as a function of SINR threshold $T$ for different significant wave heights $H_{\textrm{s}}$.}
	\label{fig:cp_vs_T} 
	\vspace*{-4mm}
\end{figure}
\begin{figure}[!b]
	\vspace*{-3mm}
	\begin{center}
		\includegraphics[width=1\columnwidth]{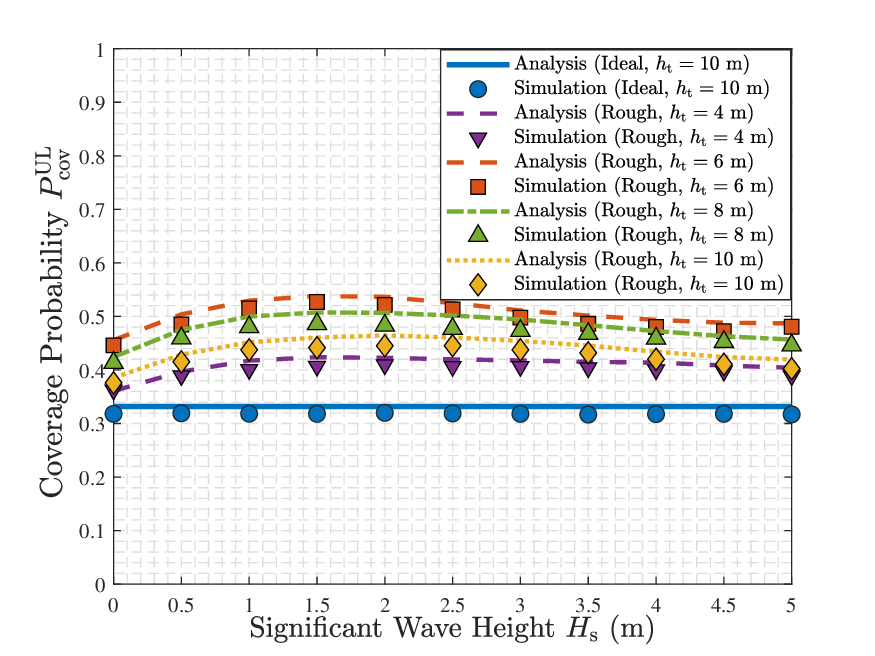}
	\end{center}
	\vspace*{-6mm}
	\caption{Coverage probability as a function of significant wave height $H_{\textrm{s}}$ for different transmitter antenna heights $h_{\textrm{t}}$.}
	\label{fig:core_finding} 
\end{figure}

Fig.~\ref{fig:cp_vs_T} plots the coverage probability against the SINR threshold $T$, revealing a critical performance crossover. This phenomenon underscores the dual role of sea roughness. For reliability-centric services that require coverage at low SINR thresholds, such as safety messages, moderate roughness is advantageous. It disrupts coherent destructive interference, thereby mitigating the probability of deep channel fades. Conversely, for rate-centric services needing coverage at high SINR thresholds, the same scattering effect is detrimental because it tempers the constructive interference peaks required to achieve high SINR values. Our model correctly captures the fact that an idealized smooth sea assumption would lead to pessimistic reliability estimates but optimistic peak-rate estimates.

This dual effect precipitates the non-monotonic behavior shown in Fig.~\ref{fig:core_finding}, which is the central finding of our work. For a fixed reliability target with $T=-15$\,dB, the coverage probability does not degrade monotonically with increasing roughness, but instead initially improves to a peak at a moderately rough sea state before declining. This phenomenon arises from a fundamental trade-off between two competing physical mechanisms. Initially, for low $H_s$, the dominant effect is an advantageous null mitigation, namely, the disruption of perfect coherent destructive interference mitigates deep, performance-killing fading nulls, thereby enhancing link robustness and increasing coverage probability. As the sea state becomes rougher, however, a detrimental power scattering effect begins to dominate. A larger fraction of the incident energy is scattered diffusely rather than coherently reflected, which reduces the overall useful reflected power and causes the average SINR to decrease. The apex of each curve thus signifies an peak coverage region under the adopted propagation model where the gains from null mitigation are maximally leveraged before the deleterious effects of specular power reduction become the dominant factor.

\section{Conclusion}\label{S6}

This paper developed a sea-state-aware stochastic geometry framework for maritime communication networks by incorporating surface roughness into physical propagation modeling. By deriving a roughness-aware reflection coefficient parameterized by significant wave height, the proposed model provides a more realistic characterization of maritime multipath propagation. The analytical results and measurement-based validation demonstrate that under the adopted roughness-aware propagation abstraction, sea roughness can have a non-monotonic impact on network performance. Moderate roughness may improve reliability-oriented coverage by alleviating destructive interference nulls, whereas stronger roughness attenuates coherent reflection and reduces high-SINR performance.These findings highlight the importance of considering realistic sea conditions in the performance analysis of maritime communication networks.  This work focuses on the physical propagation layer of non-stationary maritime networks. Together with our recent studies on fluid-spatiotemporal stochastic geometry and its extensions to network-level resource provisioning and transport-layer optimization \cite{Dong2026FSTSG,DongDigitalTides2026,Dong2026VDR}, this study represents a broader effort toward developing analytical frameworks for future non-stationary wireless systems.

\small
\bibliographystyle{IEEEtran}

\end{document}